\documentclass[12pt]{article}

\usepackage{amsmath}
\usepackage{amsfonts}
\usepackage{amssymb}
\usepackage{fullpage}
\usepackage{graphicx}
\usepackage{float}
\usepackage{complexity}
\usepackage{pgf}
\usepackage{tikz}
\usetikzlibrary{arrows,automata}
\usepackage[latin1]{inputenc}
\usepackage{verbatim}
\usepackage{caption}
\usepackage{subcaption}
\newtheorem{theorem}{Theorem}
\newtheorem{corollary}[theorem]{Corollary}
\newtheorem{lemma}[theorem]{Lemma}

\newtheorem{proposition}[theorem]{Proposition}
\newtheorem{definition}[theorem]{Definition}
\newtheorem{claim}[theorem]{Claim}

\newcommand{\qed}{\rule{7pt}{7pt}}
\newenvironment{proof}{\noindent{\bf Proof}\hspace*{1em}}{\hfill\qed\bigskip}
\newenvironment{proof-sketch}{\noindent{\bf Sketch of Proof}\hspace*{1em}}{\qed\bigskip}
\newenvironment{proof-idea}{\noindent{\bf Proof Idea}\hspace*{1em}}{\qed\bigskip}
\newenvironment{proof-of-lemma}[1]{\noindent{\bf Proof of Lemma #1}\hspace*{1em}}{\qed\bigskip}
\newenvironment{proof-attempt}{\noindent{\bf Proof Attempt}\hspace*{1em}}{\qed\bigskip}

\newcommand{\hyphen}{\lower-.12em\hbox{\textrm{-}}}
\newcommand{\rev}[1]{\ensuremath{R^{\bullet}(#1)}}
\newcommand{\DT}[1]{\ensuremath{DT(#1)}}
\newcommand{\RM}[1]{\ensuremath{RM(#1)}}
\newcommand{\vrev}[1]{\ensuremath{R^{\phi}(#1)}}
\newcommand{\contr}[1]{\ensuremath{c(#1)}}
\newcommand{\erank}[1]{\ensuremath{\chi^{\prime}_{e}(#1)}}
\newcommand{\revp}{\lang{TREE\hyphen PEBBLE}}
\newcommand{\vrevp}{\lang{TREE\hyphen VISITING\hyphen PEBBLE}}

\title{Pebbling Meets Coloring : Reversible Pebble Game On Trees}
\author{Balagopal Komarath\footnote{Sponsored by TCS Research Fellowship}
\hspace{1cm} Jayalal Sarma \hspace{1cm} Saurabh Sawlani \\[3mm]
{\large Department of Computer Science \& Engineering,} \\ 
{\large Indian Institute of Technology Madras, Chennai - 36, India.} \\[2mm]
{\large Email : \{{\tt baluks|jayalal.sarma|saurabh.sawlani\}@gmail.com}}}

\begin{document}
\maketitle

\begin{abstract}
  The reversible pebble game is a combinatorial game played on rooted
  DAGs. This game was introduced by Bennett \cite{Ben89} motivated by
  applications in designing space efficient reversible
  algorithms. Recently, Siu Man Chan \cite{Chan13} showed that the
  reversible pebble game number of any DAG is the same as its
  Dymond-Tompa pebble number and Raz-Mckenzie pebble number.
  
  We show, as our main result, that for any rooted directed tree $T$,
  its reversible pebble game number is always just one more than the
  edge rank coloring number of the underlying undirected tree $U$ of
  $T$.  The most striking implication of this result is that the
  reversible pebble game number of a tree does not depend upon the
  direction of edges, a fact that does not hold in general for DAGs.
  It is known that given a DAG $G$ as input, determining its
  reversible pebble game number is \PSPACE-hard. Our result implies
  that the reversible pebble game number of trees can be computed in
  polynomial time as edge rank coloring number of trees can be
  computed in linear time (\cite{Lam98optimaledge}).

  We also address the question of finding the number of steps required
  to optimally pebble various families of trees.  It is known that
  trees can be pebbled in $n^{{O(\log(n))}}$ steps where $n$ is the
  number of nodes in the tree. Using the equivalence between
  reversible pebble game and the Dymond-Tompa pebble game
  \cite{Chan13}, we show that complete binary trees can be pebbled in
  $n^{{O(\log\log(n))}}$ steps, a substantial improvement over the
  naive upper bound of $n^{{O(\log(n))}}$.
  
  It remains open whether complete binary trees can be pebbled in
  polynomial number of steps (i.e., $n^{k}$ for some constant
  $k$). Towards this end, we show that {\em almost optimal} (i.e.,
  within a factor of $(1+\epsilon)$ for any constant $\epsilon > 0$)
  pebblings of complete binary trees can be done in polynomial number
  of steps.

  We also show a time-space trade-off for reversible pebbling for
  families of bounded degree trees by a divide-and-conquer approach:
  for any constant $\epsilon > 0$, such families can be pebbled using
  $O(n^{\epsilon})$ pebbles in $O(n)$ steps. This generalizes an
  analogous result of Kr\'alovic\cite{Kra01} for chains.

\end{abstract}

\section{Introduction}

Pebbling games of various forms on graphs abstracts out resources in
different combinatorial models of computation (See
\cite{Chan13-thesis}). A rooted DAG can be used to model computation
as follows -- Each node in the DAG represents a value obtained during
computation, the source nodes represent input values, the internal
nodes represent intermediate values, and the root node represents the
output value. A pebble placed on a vertex in a graph corresponds to
storing the value at that node, and an edge $(a,b)$ in the graph would
represent a data-dependency - namely, the value at $b$ can be computed
only if the value at $a$ is known (or stored). Devising the rules of
the pebble game to capture the rules of the computation, and
establishing bounds for the total number of pebbles used at any point
in time, gives rise to a combinatorial approach to proving bounds on
the \emph{space} used by the computation. The Dymond-Tompa and
Raz-Mckenzie pebble games depict some of the combinatorial barriers in
improving upper bounds for depth (or parallel time) of Boolean
circuits (or parallel algorithms).

Motivated by applications in the context of reversible computation
(for example, quantum computation), Bennett\cite{Ben89} introduced the
reversible pebble game. Given any DAG $G$ with a unique sink node $r$,
the reversible pebble game starts with no pebbles on $G$ and ends with
a pebble (only) on $r$. Pebbles can be placed or removed from any node
according to the following two rules.
\begin{enumerate}
\item To pebble $v$, all in-neighbors of $v$ must be pebbled.
\item To unpebble $v$, all in-neighbors of $v$ must be pebbled.
\end{enumerate}
The goal of the game is to pebble the sink node $r$ using the minimum
number of pebbles (also using the minimum number of steps).

Recently, Chan\cite{Chan13} showed that for any DAG $G$ the number of
pebbles required for the reversible pebble game is exactly the same as
the number of pebbles required for the Dymond-Tompa pebble game and
the Raz-Mckenzie pebble game. However, connections between the
reversible pebble game and graph parameters not arising from
computational considerations were not known. For irreversible pebble
games, we know that the black white pebbling number of trees is
closely related to min-cut linear arrangements of trees\cite{Yan85}.

On the computational complexity front, Chan\cite{Chan13} also studied
the complexity of the following problem -- Given a DAG $G = (V, E)$
with a unique sink $r$ and an integer $1 \leq k \leq |V|$, check if
$G$ can be pebbled using at most $k$ pebbles. He showed that this
problem is \PSPACE-complete. Determining the irreversible black and
black-white pebbling number are known to be \PSPACE-complete on DAGs
(See \cite{Gilbert}, \cite{Hertel}). If we restrict the irreversible
black pebble game to be read-once (each node is pebbled only once),
then the problem becomes \NP-complete (See \cite{Sethi75}). However,
if we restrict our attention to trees, the irreversible black pebble
game\cite{loui} and black-white pebble game\cite{Yan85} are solvable
in polynomial time. The key insight is that the optimal
\emph{irreversible} (black or black-white) pebbling number of trees
can be achieved by read-once pebblings. Deciding whether the pebbling
number is at most $k$ for a given tree is in \NP\ since the optimal
pebbling serves as the certificate. We cannot show that determining
the reversible pebbling number is in \NP\ using the same argument as
we do not know whether the optimal value can always be achieved using
pebblings taking only polynomially many steps.

\paragraph{Our Results:} In this paper, we study the reversible pebble
game on trees. For an undirected tree $T$, the edge rank coloring
number of the tree is the minimum number of colors required to color
the edges of $T$ using integers such that for any two edges in $T$
having the same color $i$, there is at least one edge on the path
between those edges that has a higher color. We show that the
reversible pebbling number of any tree is exactly one more than the
edge rank coloring number of the underlying undirected tree. Besides,
the reversible pebbling number, another interesting parameter related
to reversible pebble game is the number of steps required to optimally
pebble the given DAG. For example, it is known that paths can be
optimally pebbled in $O(n\log n)$ steps. We show that the connection
with Dymond-Tompa pebble game can be exploited to show that complete
binary trees have optimal pebblings that take at most
$n^{{O(\log\log(n))}}$ steps. This is a significant improvement over
the previous upper bound of $n^{{O(\log(n))}}$ steps. It remains open
whether complete binary trees can be pebbled in polynomial number of
steps. Towards this end, we show that ``almost'' (within a factor of
$(1+\epsilon)$ for any constant $\epsilon > 0$) optimal pebblings of
complete binary trees can be done in polynomial number of steps. We
also generalize a time-space trade-off result given for paths by
Kr\'alovic to families of bounded degree trees showing that for any
constant $\epsilon > 0$, such families can be pebbled using
$O(n^{\epsilon})$ pebbles in $O(n)$ steps.

\paragraph{Complexity of Reversible Pebbling Number on
  Trees:} We show that the reversible pebbling number of trees along
with strategies achieving the optimal value can be computed in
polynomial time. This is obtained by combining our main result with
the linear-time algorithm given by Lam and Yue \cite{Lam98optimaledge}
for finding an optimal edge rank coloring of the underlying undirected
tree. Our proof of the main result also shows how to convert an
optimal edge rank coloring into an optimal reversible pebbling.

\section{Preliminaries}

We assume familiarity with basic definitions in graph theory, such as
those found in \cite{West}. A directed tree $T = (V, E)$ is called a
\emph{rooted directed tree} if there is an $r \in V$ such that $r$ is
reachable from every node in $T$. The node $r$ is called the root of
the tree.

An \emph{edge rank coloring} of an undirected tree $T$ with $k$
colors $\{ 1, \ldots ,k\}$ labels each edge of $T$ with a color such
that if two edges have the same color $i$, then the path between
these two edges consists of an edge with some color $j > i$. The
minimum number of colors required for an edge rank coloring of $T$
is denoted by \erank{T}.

\begin{definition}{(Reversible Pebbling\cite{Ben89})}
  Let $G$ be a rooted DAG with root $r$. A \emph{reversible pebbling
    configuration} of $G$ is a set $P \subseteq V$ (the set of pebbled
  vertices). A \emph{reversible pebbling} of $G$ is a sequence of
  reversible pebbling configurations $P = (P_{1}, \ldots ,P_{m})$ such
  that $P_{1} = \phi$ and $P_{m} = \{ r \}$ and for every
  $i, 2 \leq i \leq m$, we have

  \begin{enumerate}
  \item $P_{i} = P_{i-1} \cup \{ v \}$ or
    $P_{i-1} = P_{i} \cup \{ v \}$ and $P_{i} \neq P_{i-1}$ (Exactly
    one vertex is pebbled/unpebbled at each step).

  \item All in-neighbors of $v$ are in $P_{i-1}$.
  \end{enumerate}
  
  The number $m$ is called the time taken by the pebbling $P$.  The
  number of pebbles or space used in a reversible pebbling of $G$ is
  the maximum number of pebbles on $G$ at any time during the
  pebbling. The \emph{persistent reversible pebbling number} of $G$,
  denoted by \rev{G}, is the minimum number of pebbles required to
  persistently pebble $G$.

  A closely related notion is that of \emph{visiting} reversible pebbling, where the pebbling $P$ satisfies (1) $P_{1} = P_{m} = \phi$ and (2) there exists a $j$ such that $r \in P_{j}$. The minimum number
  of pebbles required for a visiting pebbling of $G$ is denoted by
  \vrev{T}.
\end{definition}

It is easy to see that $\vrev{G} \leq \rev{G} \leq \vrev{G} + 1$ for
any DAG $G$. 

\begin{definition}{(Dymond-Tompa Pebble Game \cite{DT85})}
  Let $G$ be a DAG with root $r$. A Dymond-Tompa pebble game is a
  two-player game on $G$ where the two players, the pebbler and the
  challenger takes turns. In the first round, the pebbler pebbles the
  root node and the challenger challenges the root node. In each
  subsequent round, the pebbler pebbles a (unpebbled) node in $G$ and
  the challenger either challenges the node just pebbled or
  re-challenges the node challenged in the previous round. The pebbler
  wins when the challenger challenges a node $v$ and all in-neighbors
  of $v$ are pebbled.

  The Dymond-Tompa pebble number of $G$, denoted \DT{G}, is the
  minimum number of pebbles required by the pebbler to win against
  an optimal challenger play.
\end{definition}

The Raz-Mckenzie pebble game is also a two-player pebble game played
on DAGs. The optimal value is denoted by \RM{G}. A definition for the
Raz-Mckenzie pebble game can be found in \cite{RM99}. Although the
Dymond-Tompa game and the reversible pebble game look quite
different. The following theorem reveals a surprising connection
between them.

\begin{theorem}{(Theorems 6 and 7, \cite{Chan13})}
  For any rooted DAG $G$, we have $\DT{G} = \rev{G} = \RM{G}$.
  \label{lem:chan:dt=rp}
\end{theorem}

\begin{definition}{(Effective Predecessor \cite{Chan13})}
  Given a pebbling configuration $P$ of a DAG $G$ with root $r$, a
  node $v$ in $G$ is called an \emph{effective predecessor} of $r$ iff
  there exists a path from $v$ to $r$ with no pebbles on the vertices
  in the path (except at $r$).
\end{definition}

\begin{lemma}{(Claim 3.11, \cite{Chan13})}
  Let $G$ be any rooted DAG. There exists an optimal pebbler strategy
  for the Dymond-Tompa pebble game on $G$ such that the pebbler always
  pebbles an effective predecessor of the currently challenged node.
  \label{lem:chan:ep}
\end{lemma}

The height or depth of a tree is defined as the maximum number of
nodes in any root to leaf path. We denote by $Ch_{n}$ the rooted
directed path on $n$ nodes with a leaf as the root. We denote by
$Bt_{h}$ the complete binary tree of height $h$. We use
$root(Bt_{h})$ to refer to the root of $Bt_{h}$. If $v$ is any node in
$Bt_{h}$, we use $left(v)$ ($right(v)$) to refer to the left (right)
child of $v$. We use $right^{i}$ and $left^i$ to refer to iterated
application of these functions. We use the notation $Ch_{i} + Bt_{h}$
to refer to a tree that is a chain of $i$ nodes where the source node
is the root of a $Bt_{h}$.

\begin{definition}
  We define the language \revp\ as the set of all tuples $(T, k)$,
  where $T$ is a rooted directed tree and k is a integer satisfying
  $1 \leq k \leq n$, such that $\rev{T} \leq k$. The language \vrevp\
  is the same as \revp\ except that the goal is to check whether
  $\vrev{T} \leq k$.
\end{definition}

In the rest of the paper, we use the term pebbling to refer to
\textit{persistent reversible pebbling} unless explicitly stated
otherwise.

\section{Pebbling meets Coloring}
In this section, we prove our main theorem which states that the
reversible pebbling number of any tree is exactly one more than the
edge rank coloring number of its underlying undirected tree. It is
helpful to think about how to solve $\revp$ in polynomial time or even
$\NP$. The first attempt would be to try and use the pebbling sequence
as a certificate that the input tree has low pebbling number. But,
this approach fails because trees are not guaranteed to have optimal
pebbling sequences of polynomial number of steps. We propose the
\emph{strategy tree} (Definition~\ref{def:strat-tree}) as a succinct
encoding of pebbling sequences. A strategy tree describes a pebbling
sequence. The key property is that for any tree, there is an optimal
pebbling sequence that can be described using a strategy tree
(Lemma~\ref{lem:rev=strat}).

\begin{definition}{(Strategy Tree)}
  Let $T$ be a rooted directed tree. If $T$ only has a single node $v$,
  then any strategy tree for $T$ only has a single node labeled
  $v$. Otherwise, we define a strategy tree for $T$ as any tree
  satisfying
  
  \begin{enumerate}
  \item The root node is labelled with some edge $e = (u, v)$ in $T$.
  \item The left subtree of root is a strategy tree for $T_{u}$ and the
    right subtree is a strategy tree for $T\setminus T_{u}$.
  \end{enumerate}
  \label{def:strat-tree}
\end{definition}

The following properties are satisfied by any strategy tree $S$ of $T
= (V, E)$.

\begin{enumerate}
\item Each node has 0 or 2 children.
\item \label{stratprop:bij} There are bijections from $E$ to internal nodes of $S$ and from
  $V$ to leaves of $S$.
\item \label{stratprop:subtree} Let $v$ be any node in $S$. Then the
  subtree $S_{v}$ corresponds to the subtree of $T$ spanned by the
  nodes labeling the leaves of $S_{v}$. If $u$ and $v$ are two nodes
  in $S$ such that one is not an ancestor of the other, then the
  subtrees in $T$ corresponding to $u$ and $v$ are vertex-disjoint.
\end{enumerate}

\begin{lemma}
  Let $T$ be a rooted directed tree. Then $\rev{T} \leq k$ if and only if there
  exists a strategy tree for $T$ of depth at most $k$.
\label{lem:rev=strat}
\end{lemma}
\begin{proof}
  We prove both directions by induction on $|T|$.  If $T$ is a single
  node tree, then the statement is trivial.

  (if) Assume that the root of a strategy tree for $T$ of depth $k$ is
  labelled by an edge $(u, v)$ in $T$. The pebbler then pebbles the
  node $u$. If the challenger challenges $u$, the pebbler follows the
  strategy for $T_{u}$ given by the left subtree of root. If the
  challenger re-challenges, the pebbler follows the strategy for
  $T\setminus T_{u}$ given by the right subtree of the root. The
  remaining game takes at most $k-1$ pebbles by the inductive
  hypothesis. Therefore, the total number of pebbles used is at most
  $k$.

  (only if) Consider an upstream pebbler that uses at most $k$
  pebbles. We are going to construct a strategy tree of depth at most
  $k$. Assume that the pebbler pebbles $u$ in the first move where
  $e = (u, v)$ is an edge in $T$. Then the root node of $S$ is
  labelled $e$. Now we have
  $\rev{T_{u}}, \rev{T \setminus T_{u}} \leq k-1$. Let the left
  (right) subtree be the strategy tree obtained inductively for
  $T_{u}$ ($T \setminus T_{u}$). Since the pebbler is upstream, the
  pebbler never places a pebble outside $T_{u}$ ($T \setminus T_{u}$)
  once the challenger has challenged $u$ (the root).
\end{proof}

We now introduce a new game called the matching game played on
undirected trees (Definition~\ref{def:contr}). This game acts as a
link between the reversible pebble game and edge rank coloring.

\begin{definition}{(Matching Game)}
  Let $U$ be an undirected tree. Let $T_{1} = U$. At each step of the
  matching game, we pick a matching $M_{i}$ from $T_{i}$ and contract
  all the edges in $M_{i}$ to obtain the tree $T_{i+1}$. The game ends
  when $T_{i}$ is a single node tree. We define the \emph{contraction
    number} of $U$, denoted \contr{U}, as the minimum number of
  matchings in the matching sequence required to contract $U$ to the
  single node tree.
  \label{def:contr}
\end{definition}

\begin{lemma}
  Let $T$ be a rooted directed tree and let $U$ be the underlying
  undirected tree for $T$. Then $\rev{T} = k+1$ if and only if $\contr{U} = k$.
  \label{lem:strat=contr}
\end{lemma}
\begin{proof}
  First, we describe how to construct a matching sequence of length
  $k$ from a strategy tree $S$ of depth $k + 1$. Let the leaves of $S$
  be the level 0 nodes. For $i \geq 1$, we define the level $i$ nodes
  to be the set of all nodes $v$ in $S$ such that one child of $v$ has
  level $i-1$ and the other child of $v$ has level at most
  $i-1$. Define $M_{i}$ to be the set of all edges in $U$
  corresponding to level $i$ nodes in $S$. We claim that $M_{1},
  \ldots ,M_{k}$ is a matching sequence for $U$. Define $S_{i}$ as the
  set of all nodes $v$ in $S$ such that the parent of $v$ has level at
  least $i+1$. Let $Q(i)$ be the statement ``$T_{i+1}$ is obtained from
  $T_{1}$ by contracting all subtrees corresponding to nodes (See
  Property~\ref{stratprop:subtree}) in $S_{i}$''. Let $P(i)$ be the
  statement ``$M_{i+1}$ is a matching in $T_{i+1}$''. We will prove
  $Q(0)$ and $Q(i) \implies P(i)$ and $(Q(i) \wedge P(i)) \implies
  Q(i+1)$. Indeed for $i = 0$, we have $Q(0)$ because $T_{1} = U$ and
  $S_{0}$ is the set of all leaves in $S$ or nodes in $T$
  (Property~\ref{stratprop:bij}). To prove $Q(i) \implies P(i)$,
  observe that the edges of $M_{i+1}$ correspond to nodes in $S$ where
  both children are in $S_{i}$. So these edges correspond to edges in
  $T_{i+1}$ (by $Q(i)$) and the fact that these edges are pairwise disjoint since no
  two nodes in $S$ have a common child).

  To prove that $(Q(i) \wedge P(i)) \implies Q(i+1)$, consider the
  tree $T_{i+2}$ obtained by contracting $M_{i+1}$ from
  $T_{i+1}$. Since $Q(i)$ is true, this is equivalent to contracting
  all subtrees corresponding to $S_{i}$ and then contracting the edges
  in $M_{i+1}$ from $T_{1}$. The set $S_{i+1}$ can be obtained from
  $S_{i}$ by adding all nodes in $S$ corresponding to edges in
  $M_{i+1}$ and then removing both children (of these newly added
  nodes) from $S_i$. This is equivalent to combining the subtrees
  removed from $S_{i}$ using the edge joining them. This is because
  $M_{i+1}$ is a matching by $P(i)$ and hence one subtree in $S_{i}$
  will never be combined with two other subtrees in $S_{i}$. But then
  contracting subtrees in $S_{i+1}$ from $T_{1}$ is equivalent to
  contracting $S_{i}$ followed by contracting $M_{i+1}$.

  We now show that a matching sequence of length at most $k$ can be
  converted to a strategy tree of depth at most $k+1$. We use proof by
  induction. If the tree $T$ is a single node tree, then the statement
  is trivial. Otherwise, let $e$ be the edge in the last matching
  $M_{k}$ in the sequence and let $(u, v)$ be the corresponding edge
  in $T$. Label the root of $S$ by $e$ and let the left (right)
  subtree of root of $S$ be obtained from the matching sequence
  $M_{1}, \ldots ,M_{k-1}$ restricted to $T_{u}$ ($T\setminus T_{u}$).
  By the inductive hypothesis, these subtrees have height at most
  $k-1$.
\end{proof}

\begin{lemma}
  For any undirected tree $U$, we have \contr{U} = \erank{U}.
  \label{lem:contr=erank}
\end{lemma}
\begin{proof}
  Consider an optimal matching sequence for $U$. If the edge $e$ is
  contracted in $M_{i}$, then label $e$ with the color $i$. This is an
  edge rank coloring. Suppose for contradiction that there exists two
  edges $e_{1}$ and $e_{2}$ with label $i$ such that there is no edge
  labelled some $j \geq i$ between them. We can assume without loss of generality that there
  is no edge labelled $i$ between $e_{1}$ and $e_{2}$ since if there
  is one such edge, we can let $e_{2}$ to be that edge. Then $e_{1}$
  and $e_{2}$ are adjacent in $T_{i}$ and hence cannot belong to the
  same matching.

  Consider an optimal edge rank coloring for $U$. Then in the
  $i^{\text{th}}$ step all edges labelled $i$ are contracted. This
  forms a matching since in between any two edges labelled $i$, there
  is an edge labelled $j > i$ and hence they are not adjacent in
  $T_{i}$.
\end{proof}

The theorems in this section are summarized in Fig.~\ref{fig:equiv}

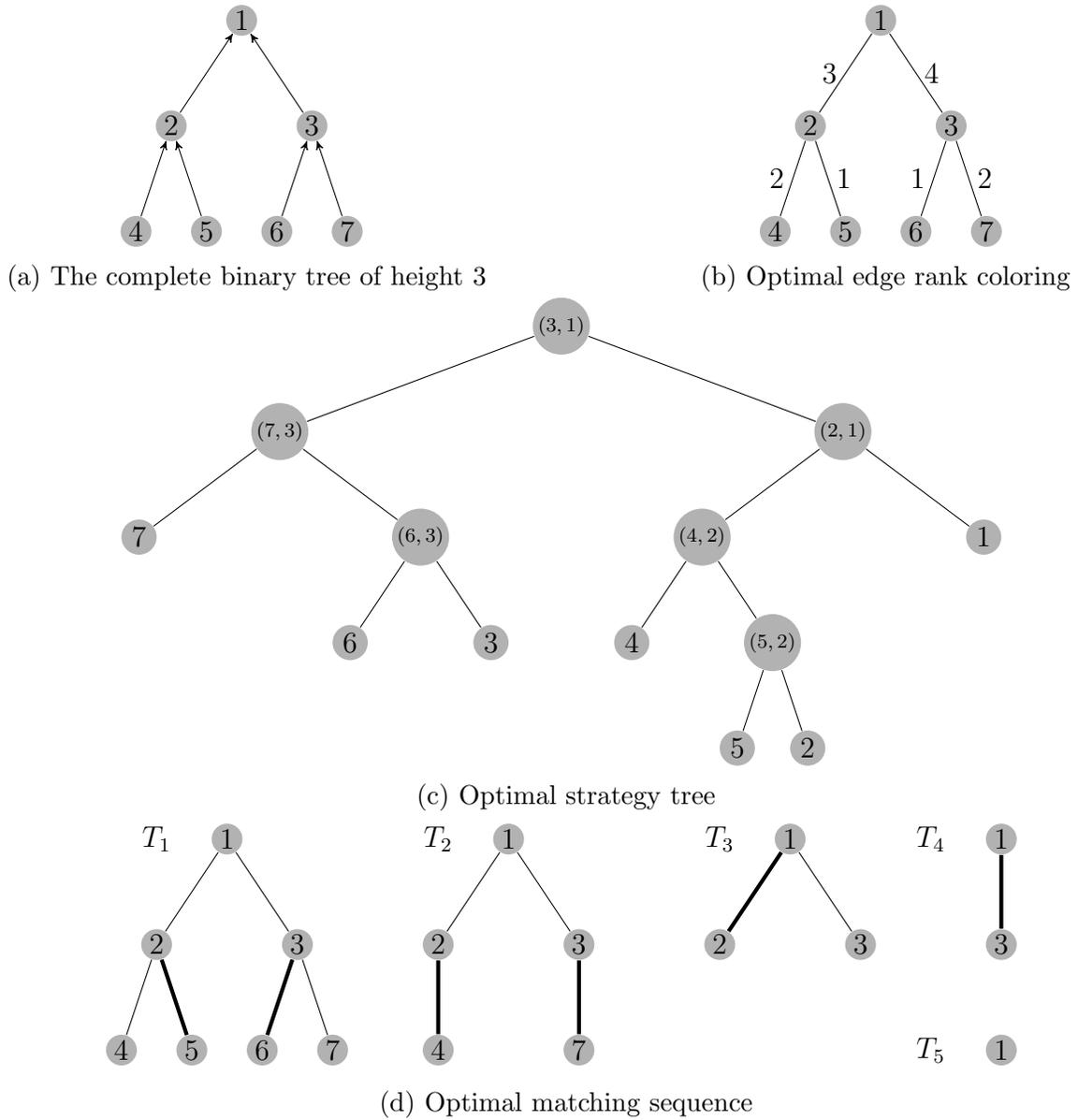
\begin{figure}
  \centering
  \begin{subfigure}{0.45\textwidth}
    \centering
    \begin{tikzpicture}
  \tikzstyle{every node}=[fill=gray!60,circle,inner sep=1pt]
  \tikzstyle{edge from parent}=[draw,<-,>=stealth']
  \tikzstyle{level 1}=[sibling distance=2cm]
  \tikzstyle{level 2}=[sibling distance=1cm]
  \node {1}
  child {node {2} child {node {4}} child {node {5}}}
  child {node {3} child {node {6}} child {node {7}}};
\end{tikzpicture}
    \caption{The complete binary tree of height 3}
    \label{fig:bt3}
  \end{subfigure}\hfill
  \begin{subfigure}{0.45\textwidth}
    \centering
    \begin{tikzpicture}
  \tikzstyle{every node}=[fill=gray!60,circle,inner sep=1pt]
  \tikzstyle{level 1}=[sibling distance=2cm]
  \tikzstyle{level 2}=[sibling distance=1cm]
  \begin{scope}
    \node {1}
    child {node {2}
      child {node {4}  edge from parent node[left,fill=none] {2}}
      child {node {5}  edge from parent node[right,fill=none] {1}}
      edge from parent node[left,fill=none] {3}}
    child {node {3}
      child {node {6}  edge from parent node[left,fill=none] {1}}
      child {node {7}  edge from parent node[right,fill=none] {2}}
      edge from parent node[right,fill=none] {4}};
  \end{scope}
\end{tikzpicture}
    \caption{Optimal edge rank coloring}
    \label{fig:erank3}
  \end{subfigure}

  \begin{subfigure}{0.9\textwidth}
    \centering
    \begin{tikzpicture}
  \tikzstyle{every node}=[fill=gray!60,circle,inner sep=1pt,minimum size=0.5cm]
  \tikzstyle{level 1}=[sibling distance=8cm]
  \tikzstyle{level 2}=[sibling distance=4cm]
  \tikzstyle{level 3}=[sibling distance=2cm]
  \tikzstyle{level 4}=[sibling distance=1cm]
  \tikzstyle{edgesz}=[font=\scriptsize]
  \node[edgesz] {$(3, 1)$}
  child {
    node[edgesz] {$(7, 3)$}
    child { node {$7$}}
    child {
      node[edgesz] {$(6, 3)$}
      child {node {$6$}}
      child {node {$3$}}
    }
  }
  child {
    node[edgesz] {$(2, 1)$}
    child {
      node[edgesz] {$(4, 2)$}
      child {node {$4$}}
      child {
        node[edgesz] {$(5, 2)$}
        child { node {$5$}}
        child { node {$2$}}
      }
    }
    child { node {$1$}}
  };
\end{tikzpicture}
    \caption{Optimal strategy tree}
    \label{fig:st3}
  \end{subfigure}

  \begin{subfigure}{0.9\textwidth}
    \centering
    \begin{tikzpicture}
  \tikzstyle{every node}=[fill=gray!60,circle,inner sep=1pt]
  \tikzstyle{level 1}=[sibling distance=2cm]
  \tikzstyle{level 2}=[sibling distance=1cm]
  \tikzstyle{matched}=[ultra thick]
  \node(11) {1}
  child {node {2} child {node {4}} child {node {5} edge from parent[matched]}}
  child {node {3} child {node {6} edge from parent[matched]} child {node {7}}};

  \node[left of=11,fill=none] {$T_1$};

  \begin{scope}[xshift=4cm]
    \node(21) {1}
    child {node {2} child {node {4} edge from parent[matched]}}
    child {node {3} child {node {7} edge from parent[matched]}};
    \node[left of=21,fill=none] {$T_2$};
  \end{scope}

  \begin{scope}[xshift=8cm]
    \node(31) {1}
    child {node {2} edge from parent[matched]}
    child {node {3}};
    \node[left of=31,fill=none] {$T_3$};
  \end{scope}

  \begin{scope}[xshift=11cm]
    \node(41) {1}
    child {node {3} edge from parent[matched]};
    \node[left of=41,fill=none] {$T_4$};
  \end{scope}

  \begin{scope}[yshift=-3cm,xshift=11cm]
    \node(51) {1};
    \node[left of=51,fill=none] {$T_5$};
  \end{scope}
\end{tikzpicture}
    \caption{Optimal matching sequence}
    \label{fig:matching3}
  \end{subfigure}
  \caption{This figure illustrates the equivalence between persistent
    reversible pebbling, matching game and edge rank coloring on trees
    by showing an optimal strategy tree and the corresponding matching
    sequence and edge rank coloring for height 3 complete binary
    tree.\label{fig:equiv}}
\end{figure}

\begin{theorem}
  Let $T$ be a rooted directed tree and let $U$ be the underlying
  undirected tree for $T$. Then we have $\rev{T} = \erank{U} + 1$.
  \label{thm:main}
\end{theorem}

\begin{corollary}
  \vrev{T} and \rev{T} along with strategy trees achieving the optimal
  pebbling value can be computed in polynomial time for trees.
  \label{cor:vis=pers}
\end{corollary}
\begin{proof}
  We show that \revp\ and \vrevp\ are polynomial time equivalent. Let
  $T$ be an instance of \revp. Pick an arbitrary leaf $v$ of $T$ and
  root the tree at $v$. By Theorem~\ref{thm:main}, the reversible
  pebbling number of this tree is the same as that of $T$. Let $T'$ be
  the subtree rooted at the child of $v$. Then we have
  $\rev{T} \leq k \iff \vrev{T'} \leq k-1$.

  Let $T$ be an instance of \vrevp. Let $T'$ be the tree obtained by
  adding the edge $(r, r')$ to $T$ where $r$ is the root of $T$. Then
  we have $\vrev{T} \leq k \iff \rev{T'} \leq k + 1$.

  The statement of the theorem follows from Theorem~\ref{thm:main} and
  the linear-time algorithm for finding an optimal edge rank coloring
  of trees\cite{Lam98optimaledge}.
\end{proof}

The following corollary is immediate from the equivalence of pebble
games (Theorem \ref{lem:chan:dt=rp}).

\begin{corollary}
  For any rooted directed tree $T$, we can compute \DT{T} and \RM{T}
  in polynomial time.
\end{corollary}

An interesting consequence of Theorem~\ref{thm:main} is that the
persistent reversible pebbling number of a tree depends only on its
underlying undirected graph. A natural question would be to ask
whether this fact generalizes to DAGs. The following proposition shows
that this is not the case.

\begin{proposition}
  There exists two DAGs with the same underlying undirected graph and
  different pebbling numbers.
\end{proposition}
\begin{proof}
Consider the following two DAGs
  \begin{figure}[ht]
    \begin{subfigure}{0.45\textwidth}
      \centering
      \begin{tikzpicture}[>=stealth', ->, node distance=2cm]
  \tikzstyle{every node}=[fill=gray!60,circle,inner sep=1pt]
  \node (1)  {$1$};
  \node (2) [below left of=1] {$2$};
  \node (4) [below right of=1] {$4$};
  \node (3) [left of=4, right of=2, below of=1] {$3$};
  \node (5) [below left of=4] {$5$};
  \node (6) [below of=4] {$6$};
  \node (7) [below right of=4] {$7$};
  
  \path (2) edge (1)
  (3) edge (1)
  (4) edge (1)
  (4) edge (3)
  (5) edge (4)
  (6) edge (4)
  (7) edge (4);
\end{tikzpicture}
      \caption{$\rev{G_1} = 5$}
    \end{subfigure}\hfill
    \begin{subfigure}{0.45\textwidth}
      \centering
      \begin{tikzpicture}[>=stealth', ->, node distance=2cm]
  \tikzstyle{every node}=[fill=gray!60,circle,inner sep=1pt]
  \node (1)  {$1$};
  \node (2) [below left of=1] {$2$};
  \node (4) [below right of=1] {$4$};
  \node (3) [left of=4, right of=2, below of=1] {$3$};
  \node (5) [below left of=4] {$5$};
  \node (6) [below of=4] {$6$};
  \node (7) [below right of=4] {$7$};
  
  \path (2) edge (1)
  (3) edge (1)
  (4) edge (1)
  (3) edge (4)
  (5) edge (4)
  (6) edge (4)
  (7) edge (4);
\end{tikzpicture}
      \caption{$\rev{G_2} = 6$}
    \end{subfigure}
  \end{figure}
DAGs $G_1$ and $G_2$ have the same underlying undirected graph and different persistent pebbling numbers.
\end{proof}

\section{Time Upper-bound for an Optimal Pebbling of Complete Binary Trees}

In this section, we improve time upper bounds for optimally
pebbling complete binary trees. It is known that the optimal pebbling
number of complete binary trees is $\log(h) + \theta(\log^{*}(h))$,
where $h$ is the height of the tree and $\log^{*}$ is the iterated
logarithmic function(\cite{Kra01}). We give an optimal pebbling of
complete binary trees that takes at most $n^{O(\log\log(n))}$ steps,
where $n$ is the number of nodes in the tree. Our pebbling is
essentially the same as in \cite{Kra01}. Our main contribution is to
show that the pebbling given in \cite{Kra01} is optimal. This proof ,
like the proof of Theorem~\ref{thm:main}, uses the equivalence between
the reversible pebble game and the Dymond-Tompa pebble game.

\begin{proposition}
  The following statements hold.
  \begin{enumerate}
  \item $\rev{Bt_h} \geq \rev{Bt_{h-1}} + 1$
  \item $\rev{Bt_{h}} \geq h + 2$ for $h \geq 3$
  \item (\cite{Ben89}) $\rev{Ch_{n}} \leq \lceil\log_{2}(n)\rceil + 1$ for all $n$
  \end{enumerate}
  \label{prop:bt-ch-rbp}
\end{proposition}
\begin{proof}
  (1) In any persistent pebbling of $Bt_{h}$, consider the earliest
  time after pebbling the root at which one of the subtrees of the
  root node has $\vrev{Bt_{h-1}}$ pebbles. At this time, there is a
  pebble on the root and there is at least one pebble on the other
  subtree of the root node. So, in total, there are at least
  $\vrev{Bt_{h-1}} + 2 \geq \rev{Bt_{h-1}} + 1$ pebbles on the tree.

  (2) Item (1) and the fact that $\rev{Bt_{3}} = 5$.\qed
\end{proof}

\begin{theorem}
  There exists an optimal pebbling of $Bt_{h}$ that takes at most
  $n^{O(\log\log(n))}$ steps.
  \label{thm:bt-time-upperbound}
\end{theorem}
\begin{proof}
  We will describe an optimal upstream pebbler in a pebbler-challenger
  game who pebbles $root(Bt_{h})$, $left(root(Bt_{h}))$,
  $left(right(root(Bt_{h})))$ and so on. In general, the pebbler
  pebbles $left(right^{i-1}(root(Bt_{h})))$ in the $i^{\text{th}}$
  step for $1 \leq i < h - \log(h)$. An upper bound on the number of
  steps taken by the reversible pebbling obtained from this game
  (which is, recursively pebble $left(right^{i-1}(root(Bt_{h})))$ for
  $0 \leq i < h - \log(h)$ and optimally pebble the remaining tree
  $Ch_{h - \log(h)} + Bt_{\log(h)}$ using any algorithm) is given
  below. Here the term ${(2h - \log(h) + 1)}^{3\log(h)}$ is an upper
  bound on the number of different pebbling configurations with
  $3\log(h)$ pebbles, and therefore an upper bound for time taken for
  optimally pebbling the tree $Ch_{h - \log(h)} + Bt_{\log(h)}$.
%
  \begin{align*}
    t(h) &\leq 2\left[t(h-1) + t(h-2) + \ldots + t(\log(h) + 1)\right] + {(2h - \log(h) + 1)}^{3\log(h)}\\
         &\leq 2ht(h-1) + {(2h - \log(h) + 1)}^{3\log(h)}\\
         &= O\left({(2h)}^{h}{(2h)}^{3\log(h)}\right)\\
         &= (\log(n))^{O(\log(n))} = n^{O(\log\log(n))}
  \end{align*}

  In the first step, the pebbler will place a pebble on
  $left(root(Bt_{h}))$ and the challenger will re-challenge the root
  node. These moves are optimal. Before the $i^{\text{th}}$ step, the
  tree has pebbles on the root and $left(right^{j}(root(Bt_{h})))$ for
  $0 \leq j < i - 1$. We argue that if $i < h - \log(h)$, placing a
  pebble on $left(right^{i-1}(root(Bt_{h})))$ is an optimal move. If
  the pebbler makes this move, then the cost of the game is
  $\max(\rev{Bt_{h_{1} - 1}}, \rev{Ch_{i} + Bt_{h_{1} - 1}}) =
  \rev{Ch_{i} + Bt_{h_{1} - 1}} \leq \rev{Bt_{h_{1} - 1}} + 1 = p$,
  where $h_{1} = h - i + 1$. Note that the inequality here is true
  when $i < h - \log(h)$ by Prop~\ref{prop:bt-ch-rbp}. We
  consider all other possible pebble placements on $i^{\text{th}}$
  step and prove that all of them are inferior.
\begin{itemize}
  \item \emph{A pebble is placed on the path from the root to
    $right^{i-1}(root(Bt_{h}))$ (inclusive)}: The challenger will
  challenge the node on which this pebble is placed. The cost of this
  game is then at least $\rev{Bt_{h_{1}}} \geq p$.

  \item \emph{A pebble is placed on a node with height less than $h_{1} -
    1$}: The challenger will re-challenge the root node and the cost of
  the game is at least $\rev{Ch_{i} + Bt_{h_{1} - 1}}$.
\end{itemize}
  The theorem follows. For completeness, the following figure represents the optimal pebbler strategy used in the proof of Theorem~\ref{thm:bt-time-upperbound} for proving time upper bounds for complete binary tree.
  \begin{figure}[ht]
    \centering
    \begin{tikzpicture}
  \tikzstyle{every node}=[shape=circle,draw];
  \tikzstyle{pebbled}=[fill=black];
  \matrix [ampersand replacement=\&,column sep=1cm,row sep=1cm, draw=none]
  {
    \& \node[pebbled](root) {};\\
    \node[pebbled](l) {}; \& \& \node(r) {};\\
    \& \node[pebbled](rl) {}; \& \& \node(rr) {};\\
    \& \& \node[pebbled](rrl) {}; \& \& \node(rrr) {};\\
    \& \& \& \node[pebbled](rrrl) {}; \& \& \node(rrrr) {};\\
  };

  \foreach \x/\y in {l/root, r/root, rl/r, rr/r, rrl/rr, rrrl/rrr, rrrr/rrr}
  \draw (\x) -- (\y);

  \draw[dotted] (rrr) -- (rr);
  \draw[dotted] (rrrl) -- (rrl);

  \foreach \x/\y/\z in {l/6cm/1cm, rl/4.65cm/0.8cm, rrl/3.3cm/0.5cm, rrrl/1.95cm/0.3cm, rrrr/1.95cm/0.3cm}
  \draw (\x.south) -- ++(-85:\y) -- ++(180:\z) -- cycle;

  \draw[<->] (root.north) ++(7cm,0) -- node[draw=none,anchor=west]{$i$} ++(-90:5.4cm);
\end{tikzpicture}
    \caption{An Optimal Pebbling for Complete Binary Trees\label{fig:bt}}
  \end{figure}
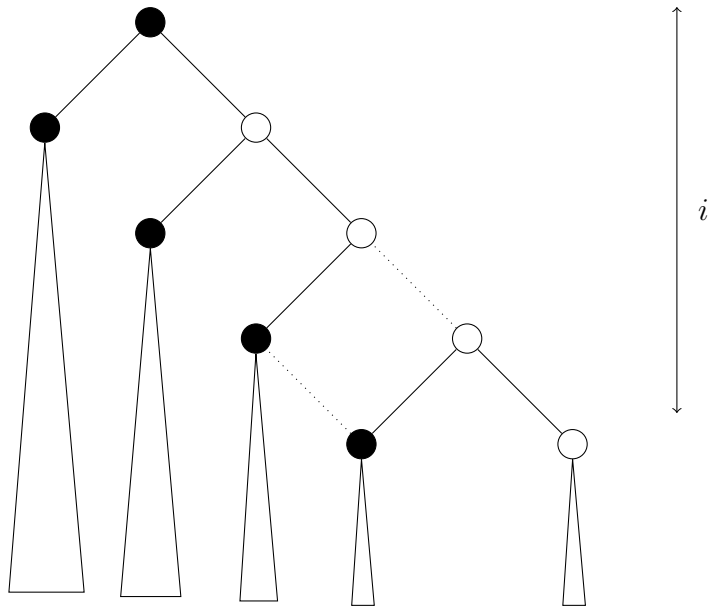  
\end{proof}

\section{Almost Optimal Pebblings of Complete Binary Trees}

In light of Theorem~\ref{thm:bt-time-upperbound}, the natural question
to ask is whether there are polynomial time optimal pebblings for
complete binary trees. In this section, we show that we can get
arbitrarily close to optimal pebblings for complete binary trees using
a polynomial number of steps (Theorem~\ref{thm:complete-trees-poly}).

\begin{theorem}
  For any constant $\epsilon > 0$, we can pebble $Bt_{h}$ using at
  most $(1 + \epsilon)h$ pebbles and $n^{O(\log(1/\epsilon))}$ steps
  for sufficiently large $h$.
  \label{thm:complete-trees-poly}
\end{theorem}
\begin{proof}
  Let $k \geq 1$ be an integer. Then consider the following pebbling
  strategy parameterized by $k$. 
  \begin{enumerate}
  \item Recursively pebble the subtrees rooted at
    $left(right^{i}(root(Bt_{h})))$ for $0 \leq i \leq k-1$ and
    $right^{k}(root(Bt_{h}))$.
  \item Leaving the $(k+1)$ pebbles on the tree (from the previous
    step), pebble the root node using an additional $k$ pebbles in $2k-1$
    steps.
  \item Retaining the pebble on the root, reverse step (1) to remove
    every other pebble from the tree.
  \end{enumerate}

  The number of pebbles and the number of steps used by the above
  strategy on $Bt_{h}$ for sufficiently large $h$ is given by the
  following recurrences.

 \[    S(h) \leq S(h-k) + (k + 1) \leq \frac{(k+1)}{k}h \]
\[     T(h)  \leq 2\left[\sum_{i = 1}^{k} T(h-i)\right] + (2k + 2) 
\leq {(2k)}^{h}(2k+2) \leq  n^{\log(k) + 1}(2k+2) \]
  where $n$ is the number of nodes in $Bt_{h}$.

  If we choose $k > 1/\epsilon$, then the theorem follows.
\end{proof}

\section{Time-space Trade-offs for Bounded-degree Trees}
In \cite{Kra01}, it is shown that there are linear time pebbling
sequences for paths that use only $n^{\epsilon}$ pebbles for any
constant $\epsilon > 0$. In this section, we generalize this result to
bounded degree trees (Theorem~\ref{thm:tree-lin-time}).

\begin{theorem}
  For any constant positive integer $k$, a bounded-degree tree $T$
  consisting of $n$ vertices can be pebbled using at most
  $O\left(n^{1/k}\right)$ pebbles and $O(n)$ pebbling moves.
  \label{thm:tree-lin-time}
\end{theorem}
\begin{proof}
  Let us prove this by induction on the value of $k$. In the base case
  ($k=1$), we are allowed to use $O(n)$ pebbles. So, the best strategy
  would to place a pebble on every vertex of $T$ in bottom-up fashion,
  starting from the leaf nodes. After the root is pebbled, we unpebble
  each node in exactly the reverse order, while leaving the root
  pebbled.

  In this strategy, clearly, each node is pebbled and unpebbled at
  most once. Hence the number of pebbling moves must be bounded by
  $2n$. Hence, a tree can be pebbled using $O(n)$ pebbles in $O(n)$ moves.

  Now consider that for $k \leq k_0 - 1$, where $k_0$ is an integer $\geq 2$, any bounded-degree
  tree $T$ with $n$ vertices can be pebbled using $O\left(n^{1/k}\right)$ pebbles in $O(n)$ moves.
  Assume that we are allowed $O\left(n^{1/k_0}\right)$ pebbles. To apply induction, we will be decomposing the tree into smaller components. We prove the following claim first. 
\begin{claim}
Let $T'$ be any bounded-degree tree with $n' > n^{(k_0-1)/k_0}$ vertices and maximum degree $\Delta$. There
exists a subtree $T''$ of $T'$ such that the number of vertices in
$T''$ is at least $\lfloor n^{(k_0-1)/k_0}/2 \rfloor$ and at most
$\lceil n^{(k_0-1)/k_0} \rceil$.
\label{lem:partition}
\end{claim}

\begin{proof}
From the classical tree-separator theorem, we know that $T'$ can be divided into two subtrees, where the larger subtree has between $\lfloor n'/2 \rfloor$ and $\lceil n' \cdot \dfrac{\Delta}{\Delta + 1} \rceil$ vertices. The key is to recursively subdivide the tree in this way and continually choose the larger subtree. However, we need to show that in doing this we will definitely strike upon a subtree with the number of vertices within the required range. Let $T_1',T_2', \ldots$ be the sequence of subtrees we obtain in these iterations. Also let $v_i$ be the number of vertices in $T_i'$ for every $i$. Note that $\forall i, \lfloor v_i/2 \rfloor \leq v_{i+1} \leq \lceil v_i \cdot \dfrac{\Delta}{\Delta+1} \rceil$. Assume that $j$ is the last iteration where $v_j > \lceil n^{(k_0-1)/k_0} \rceil$. Clearly $v_{j+1} \geq \lfloor n^{(k_0-1)/k_0}/2 \rfloor$. Also, by the definition of $j$, $v_{j+1} \leq \lceil n^{(k_0-1)/k_0} \rceil$. Hence the proof.
\end{proof}
  
The final strategy will be as follows:
\begin{enumerate}
\item Separate the tree into $\theta(n^{1/k_0})$ connected subtrees, each containing $\theta(n^{(k_0-1)/k_0})$ vertices. Claim ~\ref{lem:partition} indicates that this can always be done.

\item Let us number these subtrees in the following inductive fashion: denote by $T_1$, the `lowermost' subtree, i.e. every path to the root of $T_1$ must originate from a leaf of $T$. Denote by $T_i$, the subtree for which every path to the root originates from either a leaf of $T$ or the root of some $T_j$ for $j<i$. Also, let $n_i$ denote the number of vertices in $T_i$.
%

\item Pebble $T_1$ using $O\left(n_1^{1/(k_0-1)}\right) = O\left(n^{1/k_0}\right)$ pebbles. From the induction hypothesis, we know that this can be done using $O(n_1)$ pebbling moves.

\item Retaining the pebble on the root node of $T_1$, proceed to pebble $T_2$ in the same way as above. Continue this procedure till the root node of $T$ is pebbled. Then proceed to unpebble every other vertex by executing every pebble move up to this instant in reverse order.
\end{enumerate}
Now we argue the bounds on the number of pebbles and pebbling moves of the algorithm. Recall that the number of these subtrees is $O\left(n^{1/k_0}\right)$. Therefore, the number of intermediate pebbles at the root nodes of these subtrees is $O\left(n^{1/k_0}\right)$. Additionally, while pebbling the last subtree, $O\left(n^{1/k_0}\right)$ pebbles are used. Therefore, the total number of pebbles at any time remains $O\left(n^{1/k_0}\right)$. Each of the subtrees are pebbled and unpebbled once (effectively pebbled twice). Therefore the total number of pebbling moves is at most $\sum_i 2O(n_i)=O(n)$.
\end{proof}
%
\section{Discussion \& Open Problems}

We studied reversible pebbling on trees. Although there are polynomial
time algorithms for computing black and black-white pebbling numbers
for trees, it was unclear, prior to our work, whether the reversible
pebbling number for trees could be computed in polynomial time. We also
established that almost optimal pebbling can be done in polynomial
time.

We conclude with the following open problems. 

\begin{itemize}
\item Prove or disprove that there is an optimal pebbling for
  complete binary trees that takes at most $O\left(n^k\right)$ steps for a fixed
  $k$.
\item Prove or disprove that the there is a constant $k$ such that
  optimal pebbling for any tree takes at most $O\left(n^k\right)$ (for black and
  black-white pebble games, this statement is true with $k = 1$).
\item Give a polynomial time algorithm for computing optimal pebblings
  of trees that take the smallest number of steps.
\end{itemize}

\bibliographystyle{plain}
\bibliography{RevPebTrees.bib}

\end{document}